\documentclass[11pt]{article} 

\usepackage{fancyhdr}

\usepackage{amssymb}
\usepackage{amsmath}
\usepackage{amsthm}

\usepackage{graphicx}
\usepackage{subfigure}
\usepackage{color}

\newtheorem{Th}{Theorem}
\newtheorem{De}{Definition}
\newtheorem{Lemma}{Lemma}
\newtheorem{Cor}{Corollary}

\DeclareMathOperator{\rank}{rank}

\newcommand{\Za}{\mathcal{Z}_{\alpha}}
\newcommand{\Zaf}{\mathcal{Z}_{5\alpha}}

\newcommand{\cH}{{\mathcal H}}
\newcommand{\cZ}{\mathcal Z}

\newcommand{\R}{{\mathbb R}}
\newcommand{\Z}{{\mathbb Z}}
\newcommand{\Q}{{\mathbb Q}}
\newcommand{\C}{{\mathbb C}}

\newcommand{\gabg}{\mathcal{G}(g,\alpha,\beta)} 
\newcommand{\La}{\Lambda}
\newcommand{\la}{\lambda}

\begin{document}

\title{\bf\vspace{-39pt} Gabor frames with rational density.}
\author{YURII LYUBARSKII AND PREBEN GR{\AA}BERG NES} 
\date{}
\maketitle \thispagestyle{fancy}

\begin{abstract} We consider the frame property of the Gabor system
  $\mathcal{G}(g,\alpha,\beta) = \{ e^{2\pi i \beta n t} g(t - \alpha
  m) : m,n \in \Z\}$ for the case of rational oversampling, i.e.
  $\alpha, \beta \in \Q$. A 'rational' analogue of the Ron-Shen
  Gramian is constructed, and prove that for any odd window function
  $g$ the system $\mathcal{G}(g,\alpha,\beta)$ does not generate a
  frame if $\alpha \beta = \frac{n-1}{n}$. Special attention is paid
  to the first Hermite function $h_1(t) = t e^{-\pi
    t^2}$. \vspace{5mm}  \vspace{3mm}
 \end{abstract}

\section{Introduction}\label{S:I}
One of the fundamental problems of Gabor analysis can be stated as
follows: given a window function $g\in L^2(\R)$, determine the set of
lattice parameters $\alpha,\beta {>}0$ such that the Gabor system
$\mathcal{G}(g{,}\alpha{,}\beta) {=} \{e^{2i\pi \beta n t}
g(t{-}\alpha m){:} m{,}n {\in} \Z\}$ forms a frame in $L^2(\R)$.
\begin{equation*}
  \mathcal{F}(g):= \{(\alpha,\beta)\in \R^2_+ :
  \mathcal{G}(g,\alpha,\beta) \ \mbox{is a frame for} \ L^2(\R)\}.
\end{equation*}
\indent We remind some known facts about the frame set
$\mathcal{F}(g)$, following \cite{kG01} and (in more compressed form)
\cite{kG11}. Under milder conditions, precisely if $g$ is in the
Feichtinger algebra $M^1$ the set $\mathcal{F}(g)$ is open in $\R^2_+$
and contains a neighborhood of the origin. Fundamental density
theorems \cite{iD92,kG01,cH07} together with a version of the
uncertainty principle \cite{jB95,wC06} asserts that also
$\mathcal{F}(g)\subset \Pi_+:= \{(\alpha, \beta)\in \R^2_+ :
\alpha\beta < 1\}$ for $g\in M^1$.\\ \indent Up to the very recent
time just few functions have been known for which
$\mathcal{F}=\Pi_+$. The list included the Gaussian $g(t)=e^{-\pi t2}$
\cite{yL92,kS92a,kS92b}, the hyperbolic secant $g(t)=
(e^t+e^{-t})^{-1}$, one- and two-sided exponential functions
$g(t)=e^{-t} {\mathbf 1}_{\R_+}$ (in this case $\alpha \beta = 1$ also
generates a frame) and $g(t)=e^{-|t|}$ \cite{aJ96,aJ02}, as well as
their shifts, dilates, and Fourier transforms. A breakthrough was
achieved in \cite{kG11} where the authors constructed an infinite
family of functions for which $\mathcal{F}(g)=\Pi_+$ by proving that
any totally positive function of finite type possesses this property.\\
\indent On the other hand it is shown in \cite{aJ03} that the set
$\mathcal{F}(g)$ may have rather complicated structure even for the
"simple" function such as the characteristic function $g= {\mathbf
  1}_I$ of an interval.\\ \indent In this article we attempt to study
the set $\mathcal{F}(g)$ for some cases when $\mathcal{F}(g)\neq
\Pi_+$ {\em and} $g$ is well concentrated both in time and
frequency. Our primary objective is the first Hermite function
$h_1(t)=te^{-\pi t2}$.  This choice is motivated by the uncertainty
principle ($h_1$ minimizes the Heisenberg uncertainty among all
functions which are orthogonal to the Gaussian) and also
recent results regarding vector-valued Gabor frames \cite{kG08}.\\
\indent The article is organized as follows.  The next section
contains notation and basic facts from Gabor analysis which will be
used in the sequel. The whole analysis is carried out for the case of
rational oversampling: $\alpha\beta = p/q \in \Q$.  In Section
\ref{S:III} we prove that for {\em any} odd function $g\in M^1$ (in
particular $h_1$ of course) the set $\mathcal{F}(g)$ does NOT contain
the union of hyperbolas:
\begin{equation}\label{not}
\alpha \beta = \frac{n-1}n \  \Rightarrow \ (\alpha,\beta) \not\in
\mathcal{F}(g), \ n=2,3, \ldots \ .
\end{equation}
The proof is based on analysis of the vector-valued Zak transform (see
\cite{mZ97} and also \cite[ch.~8]{kG01}) which represents the frame
operator as matrix multiplication in a space of vector-valued
functions.  In the next section we factorize the matrix of the
vector-valued Zak transform and extract a factor which is a rational
analogue of the well-known Ron-Shen Gramian (see \cite{aR97} and also
\cite{kG01}).  In Section \ref{S:V} we conjecture that condition
\eqref{not} is the only restriction on the set $\mathcal{F}(h_1)$:
\begin{equation*}
  \mathcal{F}(h_1)= \{(\alpha,\beta)\in \Pi_+: \ \alpha\beta\neq \frac{n-1}n, \
  n=2,3, \ldots \ \}. 
\end{equation*}
Unfortunately we are not able to prove this conjecture in its full
range. We prove it analytically just for some points in $\Pi_+$ and
also provide numerical verification for a wider set of points. These
constructions are based on the rational analogue of the Gramian given
in Section \ref{S:IV}.\\ \indent The authors thank E.Malinnikova for
useful discussions and hints.
\section{Preliminaries}\label{S:II}
In this section we remind the basic facts from Gabor analysis which
will be used later.  We refer the reader to \cite{kG01} for a more
detailed presentation as well as the history of the subject.\\
\indent Given numbers $\alpha, \beta >0$ and a window function $g$ we
consider the lattice $\La= \alpha \Z \times \beta \Z$ and the Gabor
system
\begin{equation*}
  \gabg= \{ \pi_\la g: \ \la \in \La\},
\end{equation*}
where $\pi_{\la}:g \rightarrow e^{2\pi i b t}g(t-a)$ for $\la = (a,b)$
denotes the usual time-frequency shift.  The Gabor frame
operator $S_{g, \Lambda}: L^2(\R) \to L^2(\R)$ is defined as
\begin{equation*}
  S_{g, \Lambda} f(t) \,{=}\, \sum_{\lambda
    \,{\in}\, \La} \langle f,\pi_{\lambda}g \rangle_{L^{2}(\mathbb{R})}
  \pi_{\lambda}g(t), \qquad f \,{\in}\, L^2(\mathbb{R}).
\end{equation*} 
If $g$ belongs to the modulation space $M^1(\mathbb{R})$ (we remind
the definition later in this section) this operator is bounded, and
$\mathcal{G}(g,\Lambda)$ is a Gabor frame if and only if the Gabor
frame operator is invertible.  In this article we consider the case
$\alpha \beta \in \Q$. The operator $S_{g, \Lambda} $ can in this case
be realized as a
multiplication-operator in a space of vector-valued functions.\\
\indent Let $\alpha\beta = p/q$ for some relatively primes $p,q\in
\mathbb{N}$. Consider the rectangle $Q_{\alpha, \beta}= [0,\alpha/p)
\times [0, 1/\alpha)$ and the space of
vector-valued functions $\cH_{\alpha,p}= L^2(Q_{\alpha,p}, \C^p)$.\\
\indent We remind that the Zak transform is defined as
\begin{equation}\label{eq:zak} 
  \mathcal{Z}_{\alpha}f(t,\omega) = \sum_{n \in
    \mathbb{Z}} f(t - \alpha n)e^{2\pi i n \alpha \omega}.
\end{equation} 
Following \cite[ch.~8]{kG01} we consider the {\em vector-valued} Zak
transform $\overrightarrow{\cZ}_{\alpha} : L^2(\R) \to \cH_{\alpha,p}$
defined as
\begin{equation*}
  \overrightarrow{\cZ}_{\alpha}f(x,\omega)
  = \big( \cZ_\alpha(x+\frac{\alpha}{p} r, \omega) \big)_{r=1}^p, \quad (x,\omega)
  \in Q_{\alpha,p} 
\end{equation*}
The vector-valued Zak transform is up to normalization a unitary
mapping between $L^2(\R)$ and $\cH_{\alpha,p}$.\\ \indent Denote also
\begin{equation*}
  A_r^s(x,\omega)= \alpha \sum_{j=0}^{q-1} \overline {\Za
    g(x+\frac{\alpha}{p}s, \omega - \beta j)} \Za g
  (x+\frac{\alpha}{p} r,\omega - \beta j) e^{2\pi i j (r-s)/q},
\end{equation*}
and consider the $p\times p$ matrix function
\begin{equation*}
  \mathcal A (x,\omega ) = \left ( A_r^s(x,\omega) \right )_{r,s =
    0}^{p-1}, \ \ (x,\omega) \in Q_{\alpha,p}.
\end{equation*}
\noindent{\bf Theorem A}( Zibulskii-Zeevi) (see \cite{mZ97} and also
Theorem 8.3.3. in \cite{kG01}). {\em With the above assumptions we
  have }
\begin{equation*}
  \overrightarrow{\cZ}_{\alpha} (S_{g, \alpha,\beta } f)
  (x,\omega)= \mathcal A (x,\omega ) \overrightarrow{\cZ}_{\alpha}
  f(x,\omega),
\end{equation*}  
{\em for almost all }$(x,\omega)\in Q_{\alpha,p}$.
\medskip\\ In what follows we assume that $g $ belongs to the
modulation space $M^1(\mathbb{R})$.
\begin{De}\label{De:M}
  The modulation space $M^1(\mathbb{R})$ consists functions $g$ for
  which the norm
  \begin{equation*}
    \| g \|_{M_1(\mathbb{R})} = \int_{\mathbb{R}} \int_{\mathbb{R}}
    |V_fg(x,\omega)|\; dx \; d\omega < \infty
  \end{equation*}
  for some (or equivalently all) non-trivial function $f$ in the
  Schwartz space $\mathcal{S}(\R)$.
\end{De}
If $g\in M^1(\mathbb{R})$ then $\Za g$ is continuous.\medskip
\begin{Cor}
  The Gabor system $\mathcal{G}(g,\Lambda)$ is a frame in $L^2(\R)$ if
  and only if 
  \begin{equation}\label{framea} 
    \mbox{det}{\mathcal A}(x,\omega) \neq 0, \qquad (x,\omega
    )\in Q_{\alpha,p}.
  \end{equation}
\end{Cor}
\medskip We factorizes the matrix $\mathcal{A}$ in order to make
condition \eqref{framea} more transparent. Consider the column vectors
\begin{equation*}
  X^{j}(x,\omega) = \left ( X^{j}_r(x,\omega)\right
  )_{r=0}^{p-1}, \qquad \ j=0,1,\ldots \ , q-1 
\end{equation*}
where
\begin{equation}\label{column2} 
  X^{j}_r(x,\omega)= \Za g(t+\frac{\alpha r}p,
  \omega-\beta j)e^{2\pi i jr/q},
\end{equation}
and the $p\times q$ matrix
\begin{align*}
  \mathcal {Q}(t,\omega) = \left ( X^{j} \right )_{j=0}^{q-1}=\Big(
  \mathcal{Z}_{\alpha}g(t + \frac{\alpha r}{p}, \omega - \beta j)
  e^{2\pi i j r/q} \Big)_{r = 0,j=0}^{p-1,q-1}.
\end{align*}
Let $\mathcal{Q}^T$ denote the conjugate transform of
$\mathcal{Q}$. Clearly
\begin{equation*}
  \mathcal {A}(x,\omega)= \mathcal{Q}(x,\omega)
  \mathcal{Q}^T(x,\omega).
\end{equation*}
We have
\begin{equation*}
  \mathcal A (x,\omega )
  \overrightarrow{\cZ}_{\alpha} f(x,\omega) = \sum_{j=0}^{q-1}
  \langle X^{j}(x,\omega), \overrightarrow{\cZ}_{\alpha} f(x,\omega) \rangle
  \langle X^{j}(x,\omega), \overrightarrow{\cZ}_{\alpha} f(x,\omega) \rangle,
\end{equation*} 
so the condition \eqref{framea} is met if and only if for each
$(x,\omega)\in Q_{\alpha,p}$ the vectors $X^{j}(x,\omega), \
j=0,1,\ldots q-1$ span $\C^p$.\medskip
\begin{Cor}\label{cor:frameb}
  Let $\alpha \beta = \frac p q \in \Q$ and $g\in M^1(\mathbb{R})$. If
  $\mathcal{G}(g,\Lambda)$ is a frame in $L^2(\R)$ it is necessary and
  sufficient that
  \begin{equation*}
    \mbox{rank}\mathcal {Q}(x,\omega)= p, \qquad {for \ all} \
    (x,\omega) \in Q_{\alpha, p}.
  \end{equation*}
\end{Cor} 
In the next section we use this condition in order to study
Gabor systems generated by odd functions.
\section{Gabor frames generated by odd functions}\label{S:III}
In this section we prove the following
\begin{Th}\label{th:symmetry}
  Let $g\in M^1(\mathbb{R})$ be an odd function and $\alpha \beta =
  \frac {n-1} n$, $n = 2, 3, \dots $.  Then $\mathcal{G}(g,\Lambda)$
  cannot form a frame in $L^2(\R)$.
 \end{Th}
 \begin{proof} We will prove that
   \begin{equation}\label{smallrank}
     \rank \mathcal Q(0,0)< n-1.
   \end{equation}
   The results then follows from Corollary \ref{cor:frameb}.\\ \indent
   Relation \eqref{smallrank} will follow from the fact that for odd
   windows the elements of the matrix $\mathcal Q(0,0)$ posses
   additional symmetries. For simplicity we will assume $\alpha=1$.
   \begin{Lemma}\label{l:symmetry}
     Let $g\in M^1(\mathbb{R})$ be an odd function, and let
     $\alpha=1$, $\beta=p/q \in \Q$ and $ X^j_s=X^{j}_s(0,0), $ where
     the functions $X^{j}_s(x,\omega)$ are defined by \eqref{column2}.
     Then
     \begin{equation}\label{symmetry} 
       X_s^j=-X^{q-j}_{p-s}, \ s=0,1,\ldots \ , p-1, \
       j=0,1, \ldots \ , q-1.
     \end{equation}
   \end{Lemma}
   The proof of the lemma follows readily from the definition of the
   Zak transform and also from the fact that $g$ is odd.  We will use
   this lemma for $q-p=1$.\\ \indent First we consider the case
   $q=2k+2$, $p=2k+1$ ($q$ is an even number).  In this case $\mathcal
   Q (0,0)$ is a $(2k+1)\times (2k+2)$ matrix.\\ \indent We need
   additional relations for the elements of the zero row, the zero
   column and also the $(k+1)$-th column of $\mathcal Q (0,0)$.
   Namely
   \begin{align}
     X_0^0=0, \ X^{k+1}_0=0, \ X_0^j=-X_0^{q-j} \quad &\text{-- zero
       row}\label{zerorow}\\
     X^0_s=-X^0_{p-s}, \ s=1, \ldots \ ,p-1 \quad
     &\text{-- zero column};\label{zerocolumn}\\
     X^{k+1}_s=-X^{k+1}_{p-s}, \ s=1,\ldots \, p-1 \quad &\mbox{--
       $-(k+1)$th column.}\label{nplusonecolumn}
   \end{align} 
   As in Lemma \ref{l:symmetry} these relations follow readily from
   the definition of the Zak transform.\\ \indent Let $R_s$ denote the
   $s$-th row of $\mathcal Q(0,0)$. Consider the row vectors
   $e_l=(e_l^j)_{j=0,1,\ldots 2k+2}$, $l=1,2, \ldots \ ,k$, where
   $e_l^j=0$, for $j\neq l+1, 2k+2-l$, $e_l^{l+1}=1$, and
   $e_l^{2k+2-l}=-1$.\\ \indent From the relations \eqref{symmetry},
   \eqref{zerorow}, \eqref{zerocolumn}, and \eqref{nplusonecolumn} it
   is easy to see that all rows of $\mathcal Q$ belong to ${\mathcal S
     = \mbox{span} \left \{ \{R_s\}_{s=1}^k\} \cup \{ e_l\}_{l=0}^k
     \right \}}$.\\ \indent Indeed the row $R_0$ has the form
   \begin{equation}\label{form} 
     R=(0, \alpha_1, \ldots \, \alpha_k, 0, -\alpha_k,
     \ldots , \ -\alpha_1)
   \end{equation} 
   for some $\alpha_1, \ldots \ , \alpha_k$, thus the vector can be
   spanned by $ \{ e_l\}_{l=0}^n$. The rows $R_s$ $s=1,\ldots \ ,k$
   belong to the spanning set themselves. So it suffices to prove that
   the rows $R_{p-s}$, $s=1, \ldots k$ also belong to $\mathcal S$ or
   equivalently the vectors $R_s+R_{p-s}$ belong to $\mathcal S$. The
   later is evident since, according \eqref{symmetry} and
   \eqref{zerocolumn}, these vectors also have the form
   \eqref{form}.\\ \indent This completes the proof of the Theorem in
   the case $q=2k+2$, $p=2k+1$.\\ \indent Consider now the case
   $p=2k$, $q=2k+1$ ($q$ is an odd number).  Once again, in addition
   to the general relation \eqref{symmetry}, we need relations for
   selected rows and columns:
   \begin{align*}
     X_0^j=-X^{q-j}_0, \quad &\mbox{-- zero row};\\
     X_{s}^0=-X^0_{p-s} \quad &\mbox{-- zero column};\\
     X^j_k=-X_k^{q- j} \quad &\mbox{-- $k$-th row}.
   \end{align*}
   Consider now the rows $e_l=(e_l^j)_{j=0,1 \ldots 2k}$,
   $l=1,2,\ldots \, k$ with $e_l^j=0$ if $j\neq l, 2k-l+1$, $e_l^j=1$,
   and $e_l^{2k-l+1}=-1$.\\ \indent Applying the same arguments as in
   the previous case we can see that the set of $2k-1$ vectors
   $\{R_s\}_{s=1,\ldots \ , k-1} \cup \{e_l\}_{l=1, \ldots \ , k}$
   spans all rows of the matrix $\mathcal Q(0,0)$.\\ \indent This
   completes the proof of Theorem \ref{th:symmetry}.
 \end{proof} 
 \section{Factorization of the Zibulskii-Zeevi matrix}\label{S:IV}
 In this section we study the Zibulskii-Zeevi matrix $\mathcal
 Q(x,\omega)$.  Our goal is to reduce it to a simpler $p\times q$
 matrix having the same rank as $\mathcal Q$. One can consider this
 simpler matrix as an analogue of the Ron-Shen Gramian \cite{aR97} for
 rationally oversampled Gabor systems.\\ \indent We believe this
 reduction is interesting by itself, it will be also used in the next
 section in order to study Gabor frames generated by the first Hermite
 function.
 \begin{Th}\label{th:LNrep}
   Let the window function $g$ belong to $M^1(\mathbb{R})$ and
   $\alpha\beta = \frac p q \in \Q$. The system $\gabg$ forms a frame
   in $L^2(\R)$ if and only if the matrix
   \begin{equation}\label{eq:LNrep}
     \mathcal P (x,\omega)= \left (Z_{\alpha q}g(x+\frac
       \alpha p (tp+sq), \omega) \right )_{s=0\ \ t=0 }^{p-1\ q-1} 
   \end{equation} 
   has $\rank = p$ for all $(x,\omega)\in Q_{\alpha,p}$ .
 \end{Th}
 \begin{proof}
   Fix $(x,\omega)\in Q_{\alpha,p}$ and let
   \begin{small}
     \begin{equation}\label{eq:01}
       X_s^j=\Za g(x {+} \frac \alpha p s, \omega {-} \beta j) e^{2i\pi
         \frac {js}q} = \sum_n g \big(x {+} \frac \alpha p (s {-}
       pn)\big) e^{2i\pi \alpha n \omega} e^{2i\pi j (\frac s q
         {-} \alpha \beta n)}
     \end{equation} 
   \end{small}
   be the corresponding element of the matrix $\mathcal Q
   (x,\omega)$.\\ \indent For each $s=0,1,\ldots \ , p-1$,
   $t=0,1,\ldots \ , q-1$ let
   \begin{equation*}
     L(s):=\{l: l=s-pn, n\in \Z\}, \ L(s,t)=\{l\in L(s): l=t+mq, m\in
     \Z\}.
   \end{equation*}
   Setting $l=s-pn$ in \eqref{eq:01} we obtain
   \begin{align}\label{eq:02}
     X_s^j=\sum_{l\in L(s)}g
     \big(x+\frac \alpha p
     l\big)& e^{2i\pi \frac {jl}q +2i \pi \alpha \omega \frac {s-l}p}\notag\\ &=
     \sum_{t=0}^{q-1} e^{2i \pi \alpha \omega \frac s p} \sum_{l\in
       L(s,t)}{g (x+\frac \alpha p l) e^{- 2i\pi\alpha \omega
         \frac l p} } e^{2i\pi \frac {jt}q}.
   \end{align} 
   For each $t\in \{0, \ldots \, q-1\}$ and $s\in \{ 0, \ldots \
   ,p-1\}$ we chose numbers $k_t \in \{0, \ldots \, q-1\}$ and $m_s
   \in \{ 0, \ldots \ ,p-1\}$ such that
   \begin{equation*} 
     k_t p = t \mbox{(mod $q$)}, \quad m_s q = s \mbox{(mod $p$)}.
   \end{equation*}
   One can easily see that $L(s,t)=\{k_tp+m_sq - pq m : m\in \Z\}$, so
   one can rewrite \eqref{eq:02} as
   \begin{multline}
     \label{eq:03} X_s^j= e^{2i\pi \alpha \omega\frac{s-m_sq}p}
     \sum_{t=0}^{q-1} e^{2i \pi \alpha \omega k_t} e^{2i\pi k_t j
       \frac pq} \times \\ \underbrace { \sum_{m\in \Z} g\left (
         x+\frac \alpha p(k_tp +m_sq) - m\alpha q \right ) e^{2i \pi
         \omega m \alpha q} }_ {= Z_{\alpha q}g(x+\frac \alpha p
       (k_tp+m_sq), \omega)}.
   \end{multline}
   \indent Since the numbers $k_t$ runs through the set $\{0, \ldots
   \, q-1\}$ as $t$ runs through this set, we can rewrite
   \eqref{eq:03} as
   \begin{multline*}
     X_s^j= e^{2i\pi \alpha \omega\frac{s-m_sq}p} \sum_{\tau =0}^{q-1}
     e^{2i \pi \alpha \omega \tau} e^{2i\pi \tau j \frac pq} \times \\
     \underbrace { \sum_{m\in \Z} g\left ( x+\frac \alpha p(\tau p
         +m_sq) - m\alpha q \right ) e^{2i \pi \omega m \alpha q} }_
     {= Z_{\alpha q}g(x+\frac \alpha p (\tau p+m_sq), \omega)},
   \end{multline*}
   or
   \begin{equation*}
     \mathcal Q (x,\omega)= \mbox{diag}\{e^{2i\pi \alpha
       \omega\frac{s-m_sq}p} \}_{s=0}^{p-1} \ \tilde{\mathcal P}(x,\omega) \
     \mbox{diag}\{ e^{2i \pi \alpha \omega \tau} \}_{\tau=0}^{q-1} \ W,
   \end{equation*} 
   here
   \begin{equation*}
     \tilde {\mathcal P}(x,\omega)= \left ( Z_{\alpha q}g(x+\frac
       \alpha p (\tau p+m_sq), \omega) \right )_{s=0 \ \tau=0}^{p-1 \ q-1},
     \quad W= \left ( e^{2i\pi \tau j \frac pq} \right )_{\tau,j=0}^{q-1}.
   \end{equation*}
   Clearly the matrices $ \tilde {\mathcal P}(x,\omega)$ and $
   {\mathcal Q}(x,\omega)$ have the same rank. On the other hand,
   since the number $m_s$ runs through the whole set $\{0,1,\ldots \ ,
   p-1\}$ as $s$ runs through this set the matrices $ \tilde {\mathcal
     P}(x,\omega)$ and $ {\mathcal P}(x,\omega)$ differ only by
   permutations of their rows and hence have the same rank.
 \end{proof}
\section{Example}\label{S:V}
In \cite{kG08} it is proved that the Gabor system
$\mathcal{G}(h_1,\alpha,\beta)$ is a frame if $\alpha \beta <
\frac{1}{2}$, and fails to be a frame if $\alpha \beta =
\frac{1}{2}$. Furthermore the authors discuss an example suggesting
that this result may be sharp.\\ \indent In this section we prove that
for at least some $\alpha,\beta$ with $\alpha \beta > \frac{1}{2}$ the
system $\mathcal{G}(h_1,\alpha,\beta)$ indeed generates a frame in
$L^2(\R)$. The proof uses the matrix-function $\mathcal{P}$
constructed in the previous section, and also a result on
\textit{diagonally-dominant} matrices.
\begin{Th}\label{35} 
  Let $\alpha \beta = 3/5$ and $h_1(t)=te^{-\pi t^2}$.  Then the
  system $\mathcal{G}(h_1,\alpha,\beta)$ is a frame in $L^2(\R)$.
\end{Th}
\begin{proof} The matrix $\mathcal P (x,\omega)$ takes the form
  \begin{equation*}
    \mathcal P_1 (x,\omega)= \left (\Zaf h_1(x+ \alpha
      t+s\frac{5\alpha}{3}, \omega) \right )_{s,t=0 }^{2, \ 4}, 
  \end{equation*} 
  By Theorem \ref{th:LNrep} and Corollary \ref{cor:frameb} it suffices
  to prove that
  \begin{equation*} 
    \rank \mathcal P_1 (x,\omega)=3, \quad \mathrm{for\ all}\
    (x,\omega) \in Q_{\alpha,3}.
  \end{equation*} 
  \indent We split the proof into several steps.\medskip\\ {\bf a. \ }
  It suffices to prove that $\mathcal{G}(h_1,\alpha,\beta)$,
  $\alpha\beta = \frac{3}{5}$ is a frame for $\alpha \geq
  \sqrt{\frac{3}{5}}$. The case $\beta \geq \sqrt{\frac{3}{5}}$ can be
  reduced to the previous by using the Fourier transform.\medskip\\
  {\bf b. \ } Since the function $h_1$ decays fast one can approximate
  $\Zaf h_1(x,\omega)$ with the maximal term of the series.  In
  particular the following holds
  \begin{Lemma}\label{l:approximation} Let $0\leq |x| <
    \frac{5\alpha}{2}$. Then
    \begin{equation*}
      |h_1(x) - \Zaf h_1(x,\omega)| \leq C_{5\alpha}h_1(5\alpha - |x|)
    \end{equation*}
    where $C_{5\alpha} = 2 + \frac{1}{h(5\alpha)} \sum_{n\geq 2}
    h_1(5\alpha n) + h_1(\frac{5\alpha(2n - 1)}{2})$.
  \end{Lemma}
  This lemma can be verified directly. For all practical reasons we
  can assume that $C_{5\alpha} = 0$.\medskip\\ {\bf c. \ } We will see
  later in {\bf (g)} that it suffices to consider $0\leq x \leq
  \frac{\alpha}{6}$. This will follow from symmetry of $h_1$ and
  quasi-periodicity of the Zak transform.  We split the interval
  $0\leq x \leq \frac{\alpha}{6}$ in two: $0\leq x< \frac{\alpha}{12}$
  and $\frac{1}{12} \leq x \leq \frac{\alpha}{6}$.\medskip\\ {\bf d.
    \ } Let $0\leq x<\frac{\alpha}{12}$.  Consider the sub-matrix
  of $\mathcal P_1 (x,\omega)$ corresponding to $t=1,2,3$. After
  interchanging of the second and third row this matrix takes the form
  \begin{align*}
    \begin{pmatrix} 
      \Zaf h_1(x+\alpha,\omega) &
      \Zaf h_1(x+2\alpha,\omega) &
      \Zaf h_1(x+3\alpha,\omega)\\
      \Zaf h_1(x+ \frac{13\alpha}{3},\omega) & \Zaf h_1(x+
      \frac{16\alpha}{3},\omega) & \Zaf h_1(x+
      \frac{19\alpha}{3},\omega) \\
      \Zaf h_1(x+ \frac{8\alpha}{3},\omega) &
      \Zaf h_1(x+\frac{11\alpha}{3},\omega) &
      \Zaf h_1(x+\frac{14\alpha}{3},\omega)
    \end{pmatrix}.
  \end{align*}    
  We will use the following theorem about diagonally dominant
  matrices.\medskip\\
  \noindent{\bf Theorem B} (see \cite{oT49}). {\em If $(a_i^k)$ is an
    $n\times n$-matrix with complex elements such that either}
  \begin{itemize}
    \item[(i):] \qquad$|a_i^i| > \sum_{k,k\neq i}|a_k^i|, \qquad
      \qquad \qquad \qquad \qquad \;\;\; 1 \leq i \leq n$
    \end{itemize}
    {\em or}
    \begin{itemize}
    \item[(ii):] \qquad$|a_{i}^i||a_{j}^j| > \Big(\sum_{k,k\neq
        i}|a_{k}^i|\Big) \Big(\sum_{k,k\neq j}|a_{k}^j|\Big),\qquad
      1\leq i,j \leq n,\;i\neq j$
    \end{itemize}
    {\em then $\mathrm{det}(a_{ik}) \neq 0$.}
  \medskip\\
  \noindent{\bf e. \ } We will check that Theorem B can be used to
  prove invertibility of the matrix in \textbf{(d)}. We emphasis the
  main steps of the proof.
  \begin{itemize}
  \item Using the fact that $|\Zaf h_1(x,\omega)|$ is
    $5\alpha$-periodic function one can represent the absolute values
    of the elements of the matrix in \textbf{(e)} as
    \begin{align*}
  \begin{pmatrix} 
    |\Zaf h_1(x + \alpha,\omega)| & |\Zaf h_1(x +
    2\alpha,\omega)| & |\Zaf h_1(x - 2\alpha,\omega)|\\
    |\Zaf h_1(x - \frac{2\alpha}{3},\omega)| & |\Zaf h_1(x
    + \frac{\alpha}{3},\omega)| & |\Zaf h_1(x +
    \frac{4\alpha}{3},\omega)|\\ |\Zaf h_1(x -
    \frac{7\alpha}{3},\omega)| & |\Zaf h_1(x -
    \frac{4\alpha}{3},\omega)| & |\Zaf h_1(x -
    \frac{\alpha}{3},\omega)|
  \end{pmatrix}.
\end{align*}
\item Using Lemma \ref{l:approximation} one can replace all elements,
  except those with indices $(1,2), (1,3)$ and $(3,1)$ with the main
  term of the corresponding series. We use an upper bound for the
  remaining (non-diagonal) elements replacing $\Zaf h_1(x -
  \frac{7\alpha}{3},\omega)$ and $\Zaf h_1(x \pm 2\alpha,\omega)$ with
  respectively $3h_1(x - \frac{7\alpha}{3})$ and $3h_1(x \pm
  2\alpha)$. Namely, for $x$ near $\frac{5 \alpha}{2}$ Lemma
  \ref{l:approximation} gives
  \begin{align*}
    |\Zaf h_1(x,\omega)| {=} |h_1(x) {-} \Zaf h_1(x,\omega) {-}
    h_1(x)| {\leq} h_1(x) {+} C_{5\alpha} h_1(5\alpha {-} |x|).
  \end{align*}
  Since the constant $C_{5\alpha} \approx 2$ we use the estimate
  $|\Zaf h_1(x,\omega)| \leq 3h_1(x)$ for $x$ near
  $\frac{5\alpha}{2}$. Since we plan to use Theorem B we use an upper
  estimate for the non-diagonal elements. For the other values of $x$
  the error is negliable.\\ \indent \quad We will show that the
  following matrix satisfies the conditions of Theorem B.
  \begin{align*}
    \Big(H_i^j\Big)_{i,j=1}^3 =
    \begin{pmatrix} |h_1(x + \alpha)| & 3|h_1(x +
      2\alpha)| & 3|h_1(x - 2\alpha)|\\ |h_1(x -
      \frac{2\alpha}{3})| & |h_1(x + \frac{\alpha}{3})|
      & |h_1(x + \frac{4\alpha}{3})|\\ 3|h_1(x -
      \frac{7\alpha}{3})| & |h_1(x - \frac{4\alpha}{3})|
      & |h_1(x - \frac{\alpha}{3})|
    \end{pmatrix}.
  \end{align*}
\item For $0 \leq x \leq \frac{\alpha}{12}$ and $\sqrt{\frac{3}{5}}
  \leq \alpha \leq 1$ condition (ii) in Theorem B can now be verified
  by direct inspection.
\item For $0 \leq x \leq \frac{\alpha}{12}$ and $\alpha \geq 1$ we
  consider the difference 
  \begin{align*}
    H_i^i - \sum_{k,k\neq i} H_i^k,\qquad \qquad 1 \leq i \leq 3.
  \end{align*}
  If this expression is positive then condition (i) in Theorem B
  is met. Consider first the case $i=1$. Then we need to verify
  \begin{equation}\label{eq:H1}
    H_1^1 - H_1^2 - H_1^3 > 0.
  \end{equation}
  Let $\alpha y = x$ for $0 \leq y \leq \frac{1}{12}$. Then \eqref{eq:H1}
  is equivalent to
  \begin{align*}
    &h_1(\alpha (y + 1)) - 3h_1(\alpha (y + 2)) - 3h_1(\alpha (2 - y))\\ &>
    h_1(\alpha (\frac{1}{12} + 1)) - 3h_1(2 \alpha ) - 3h_1(\alpha (2 -
    \frac{1}{12}))\\ &> h_1(\frac{13\alpha}{12}) -
    6h_1(\frac{23\alpha}{12}) = h_1(\frac{13\alpha}{12})\Big(1 -
    6\frac{23}{13}e^{-\frac{5\pi \alpha^2}{2}}\Big)
  \end{align*}
  which clearly is positive for all values $\alpha \geq 1$. The proof
  of the cases $i=2$ and $i=3$ follows along the same lines.
\end{itemize}
\medskip {\bf f. \ } The case $1/12\leq x \leq 1/6$ and $\alpha \geq
\sqrt{\frac{3}{5}}$ can be treated similarly to the previous case. In
this case one can verify that the submatrix of $\mathcal{P}_1
(x,\omega)$ corresponding to the columns $t=0,2,3$ fulfill condition
$(i)$ of
Theorem B.\medskip\\
{\bf g. \ } We prove that the case $\frac{\alpha}{6}\leq x \leq
\frac{\alpha}{3}$ can be reduced to the previous ones.  By
substituting $x = \frac{\alpha}{3} - y$ we have
\begin{align*}
  |\Zaf  h_1(\frac{\alpha}{3}-y + \alpha t +
  r\frac{5\alpha}{3},\omega)| &= |\Zaf  h_1(\frac{\alpha}{3}-y +
  \alpha t + \alpha(r - 1)\frac{5}{3} + \frac{5\alpha}{3},\omega)|\\
  &= |\Zaf h_1(y - \alpha(t + 2) - \alpha(r-1)\frac{5}{3},\omega)|.
\end{align*} 
It is clear that the two cases are - up to permutations of
rows/columns - similar for $0 \leq x \leq \frac{\alpha}{6}$ and
$\frac{\alpha}{6} \leq x \leq \frac{\alpha}{3}$.
\end{proof}
\section{Conjecture}\label{S:VI}
Calculations of the previous section become too cumbersome for
arbitrary $\alpha$,$\beta$ with $\alpha \beta \in \Q$, $\alpha \beta
<1$. Numerous numerical calculations make
us to believe that the following statement is true:\medskip\\
\noindent {\bf Conjecture\quad} {\em Let $\alpha \beta <1$ and $\alpha
  \beta \neq (n-1)/n$, $n=1,2,\ldots $,. Then
  $\mathcal{G}(h_1,\alpha,\beta)$ in a frame in $L^2(\R)$. }\medskip\\
Unfortunately the authors are at the moment not able to
prove/disprove this statement even in the case $\alpha \beta \in
\Q$.\\ \indent The following figure suggest that $\alpha\beta =
\frac{n-1}{n}$ are the only exceptional cases. Figure~\ref{fig:eig}
shows the minimal eigenvalue of
$\mathcal{P}(x,\omega)\mathcal{P}^T(x,\omega)$ for $0 \leq x \leq
\frac{\alpha}{2p}$ and $0 \leq \omega \leq \frac{1}{\alpha}$ for
$\alpha \beta = \frac{n-j}{n}$ for $5 \leq n \leq 201$ and $1 \leq j
\leq n-1$. With this choice of $n$ and $j$ we have $\alpha \beta <
0.995$. In the experiments we have considered $\alpha = 1$.
\begin{figure}[!h]
  \begin{center}
    \subfigure[]{\includegraphics[height=3cm,width=10cm]{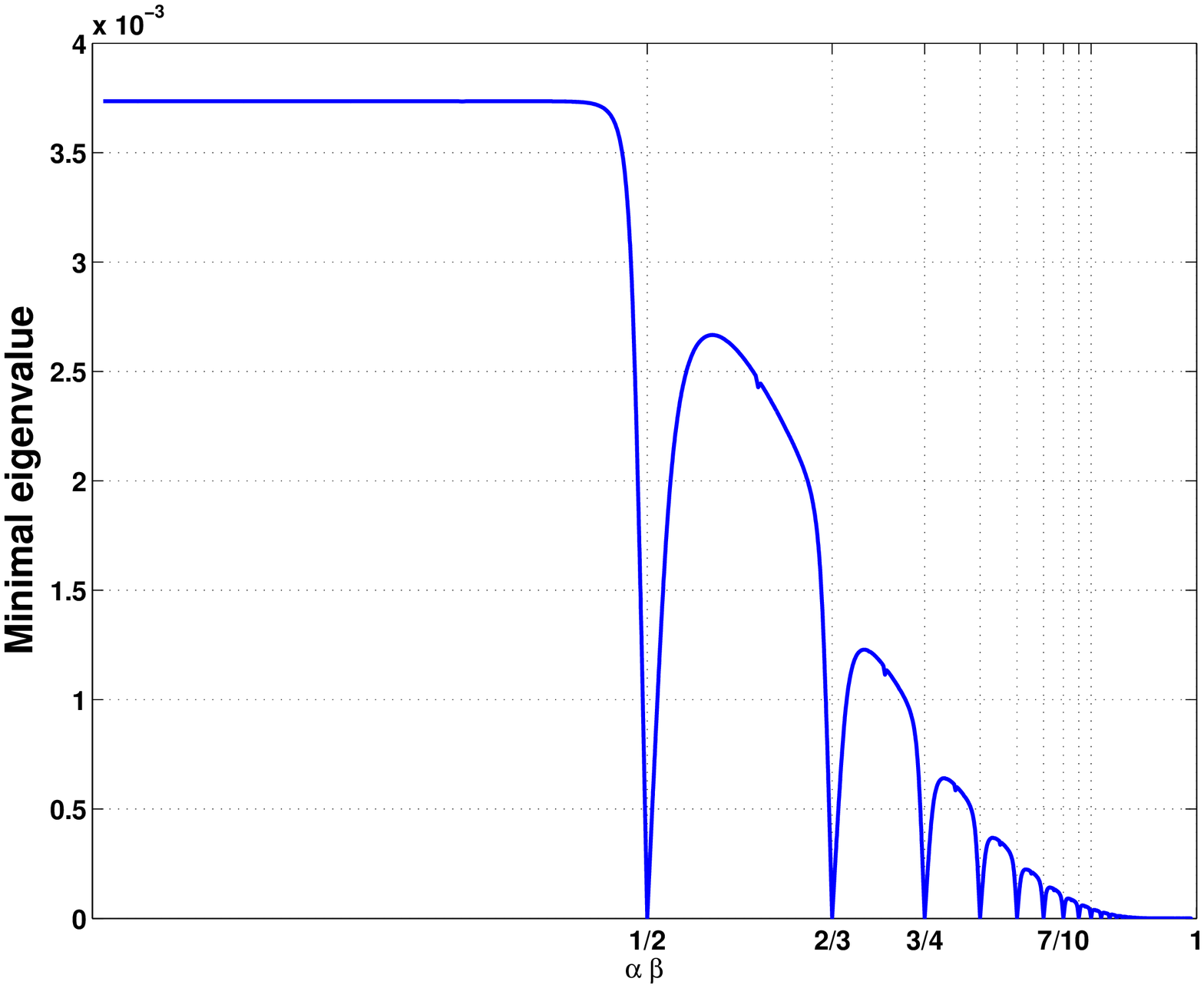}}\\
    \subfigure[]{\includegraphics[height=3cm,width=10cm]{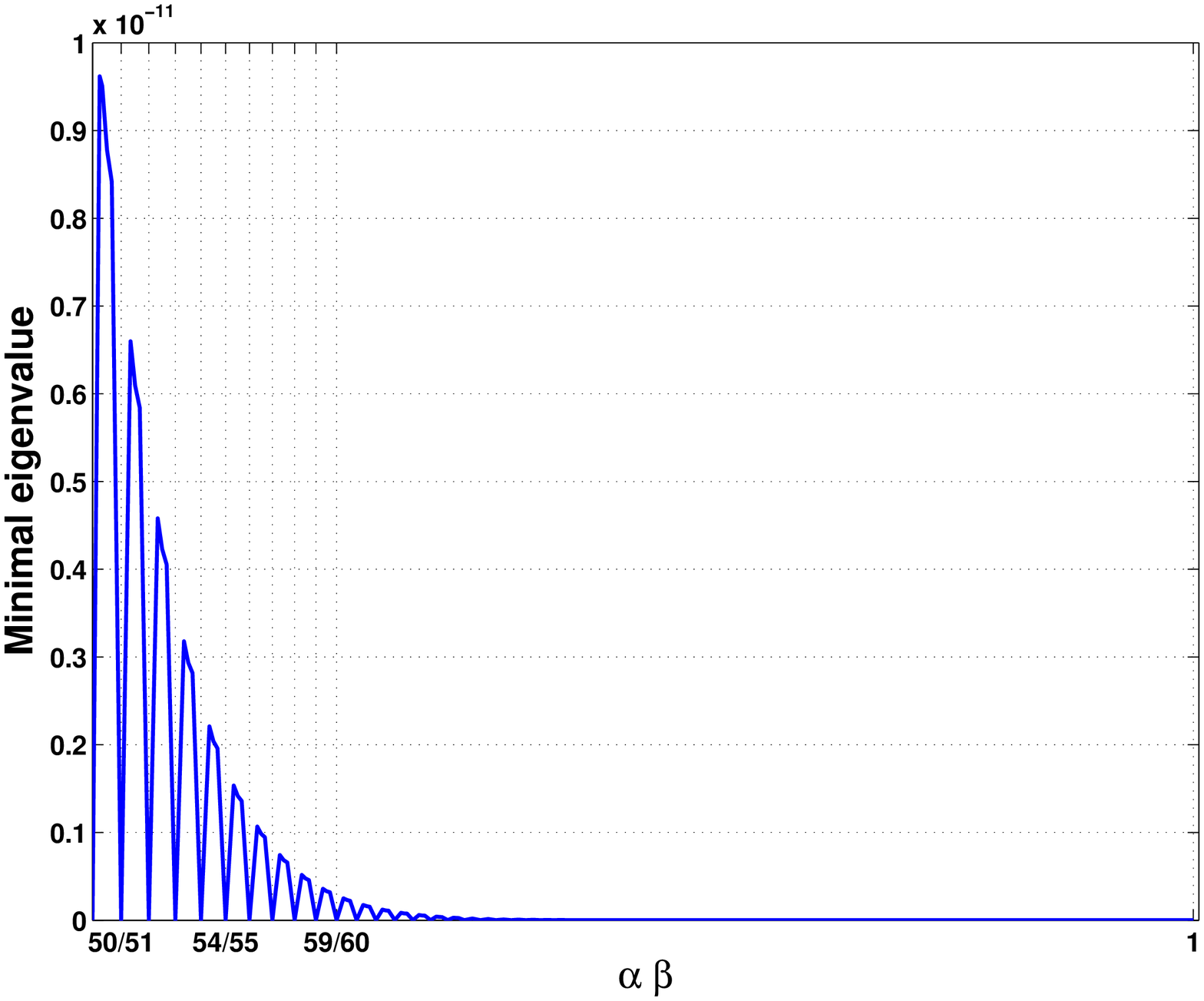}}
  \end{center}
  \caption{Minimal eigenvalue of $\mathcal{P}\mathcal{P}^T$ for (a)
    $\alpha \beta < 0.98$ and (b) $0.98 < \alpha \beta < 0.995$. In
    the experiments we have used $\alpha = 1$. Note that the scaling
    of the y-axis differs.}\label{fig:eig}
\end{figure}
\newpage
\bibliographystyle{plain}
\bibliography{bibgabor,biboperator}

Department of Mathematical Sciences, Norwegian University of Science
and Technology, N-7491 Trondheim, Norway.\\
\textit{E-mail address:} \texttt{Yurii.Lyubarskii@math.ntnu.no}\\
\textit{E-mail address:} \texttt{prebengr@math.ntnu.no}
\end{document}